\documentclass[12pt]{article}
\usepackage[authoryear]{natbib}
\usepackage{amssymb}
\usepackage{amsfonts}
\usepackage{amsmath}
\usepackage{dsfont}
\usepackage{soul}
\usepackage[nohead]{geometry}
\usepackage{graphicx}
\usepackage{amsthm}
\usepackage{color}
\usepackage{comment}
\usepackage{setspace}
\usepackage{framed}
\usepackage{enumitem}
\usepackage{todonotes}
\usepackage{tikz}
\usepackage{rotating}
\usepackage{subfigure}
\usepackage[flushleft]{threeparttable}
\usepackage{multirow}
\usepackage[colorlinks=true,citecolor=blue,linkcolor=blue]{hyperref}
\usepackage{xr}
\usepackage{bm}
\usepackage{xcolor}
\usepackage{booktabs}
\usepackage{lscape}
\usepackage{algorithm}
\usepackage{algorithmic}
\usepackage[toc,title,page]{appendix}
\usepackage{lscape}
\usepackage{pdflscape}
\usepackage{threeparttable} % For table footnotes
\usepackage{booktabs} % For better-looking tables
 
\usepackage{natbib}

\renewcommand{\today}{\ifcase \month \or January\or February\or March\or %
	April\or May\or June\or July\or August\or September\or October\or November\or %
	December\fi, \number \year} 

\usepackage{caption} 
\allowdisplaybreaks

7
\newcommand{\bff}{\mbox{\bf f}}
\newcommand{\bx}{\mbox{\bf x}}

\newcommand{\bg}{\mbox{\bf g}}
\newcommand{\bA}{\mbox{\bf A}}
\newcommand{\ba}{\mbox{\bf a}}
\newcommand{\bu}{\mbox{\bf u}}
\newcommand{\bv}{\mbox{\bf v}}
\newcommand{\bB}{\mbox{\bf B}}

\newcommand{\bG}{\mbox{\bf G}}
\newcommand{\bH}{\mbox{\bf H}}

\newcommand{\bP}{\mbox{\bf P}}
\newcommand{\bM}{\mbox{\bf M}}
\newcommand{\bW}{\mbox{\bf W}}
\newcommand{\bI}{\mbox{\bf I}}
\newcommand{\bJ}{\mbox{\bf J}}
\newcommand{\bV}{\mbox{\bf V}}
\newcommand{\bQ}{\mbox{\bf Q}}
\newcommand{\bR}{\mbox{\bf R}}

\newcommand{\bX}{\mbox{\bf X}}
\newcommand{\bY}{\mbox{\bf Y}}

\newcommand{\bmu}{\mbox{\boldmath $\mu$}}

\newcommand{\bgamma}{\mbox{\boldmath $\gamma$}}
\newcommand{\bGamma}{\mbox{\boldmath $\Gamma$}}

\newcommand{\bXi}{\mbox{\boldmath $\Xi$}}
\newcommand{\bLambda}{\mbox{\boldmath $\Lambda$}}
\newcommand{\bTheta}{\mbox{\boldmath $\Theta$}}
\newcommand{\bPi}{\mbox{\boldmath $\Pi$}}

\newcommand{\bSigma}{\mbox{\boldmath $\Sigma$}}
\newcommand{\bOmega}{\mbox{\boldmath $\Omega$}}

\newcommand{\bPhi}{\mbox{\boldmath $\Phi$}}
\newcommand{\bPsi}{\mbox{\boldmath $\Psi$}}

\newcommand{\tr}{\mathrm{tr}}
\newcommand{\diag}{\mathrm{diag}}

\newcommand{\bw}{\mbox{\bf w}}

\newcommand{\beq}{\begin{eqnarray*}}
\newcommand{\eeq}{\end{eqnarray*}}

\usepackage{graphicx,rotating,booktabs,threeparttable}

\usepackage{adjustbox}
\usepackage{array}

\newcolumntype{R}[2]{%
	>{\adjustbox{angle=#1,lap=\width-(#2)}\bgroup}%
	l%
	<{\egroup}%
}
\newtheorem{thm}{Theorem}[section]

\newtheorem{lem}{Lemma}[section]

\newtheorem{assum}{Assumption}[section]
\newtheorem{pro}{Proposition}[section]
\numberwithin{equation}{section}
\theoremstyle{definition}

\newtheorem{remark}{Remark}[section]
\makeatletter
\def\@biblabel#1{\hspace*{-\labelsep}}
\makeatother
\DeclareMathOperator*{\argmax}{argmax}

\numberwithin{equation}{section}

\renewcommand{\hat}{\widehat}

\renewcommand{\hat}{\widehat}

\newcommand{\bfm}[1]{\ensuremath{\mathbf{#1}}}

\def\ba{\bfm a}   \def\bA{\bfm A}  
   \def\bB{\bfm B}  
     
\def\bd{\bfm d}     
\def\be{\bfm e}     
\def\bff{\bfm f}    
\def\bg{\bfm g}   \def\bG{\bfm G}  
\def\bh{\bfm h}   \def\bH{\bfm H}  
   \def\bI{\bfm I}  
   \def\bJ{\bfm J}

\def\bm{\bfm m}   \def\bM{\bfm M}

   \def\bP{\bfm P}  
\def\bq{\bfm q}   \def\bQ{\bfm Q}  
   \def\bR{\bfm R}  \def\RR{\mathbb{R}}

\def\bu{\bfm u}     
\def\bv{\bfm v}   \def\bV{\bfm V}  
\def\bw{\bfm w}   \def\bW{\bfm W}  
\def\bx{\bfm x}   \def\bX{\bfm X}  
   \def\bY{\bfm Y}

 \def\cA{{\cal  A}}

 \def\cE{{\cal  E}}
 \def\cF{{\cal  F}}

 \def\cI{{\cal  I}}

 \def\cM{{\cal  M}}

 \def\cS{{\cal  S}}

\def\calX{{\cal  X}} 
 \def\cY{{\cal  Y}}

\newcommand{\bfsym}[1]{\ensuremath{\boldsymbol{#1}}}

 \def\bgamma{\bfsym \gamma}             \def\bGamma{\bfsym \Gamma}
            \def\bDelta {\bfsym {\Delta}}
               
 \def\bmu{\bfsym {\mu}}                 
 
            \def\bPsi  {\bfsym {\Theta}}
           
 \def\bsigma{\bfsym \sigma}             \def\bSigma{\bfsym \Sigma}
         \def\bLambda {\bfsym {\Lambda}}
           \def\bOmega {\bfsym {\Omega}}

 			\def  \bXi{\bfsym{\Xi}}
        
 \def\bPsi{\bfsym {\Psi}}
\def\bvartheta{\bfsym{\vartheta}}	
\def\bTheta{\bfsym{\Theta}}	
 
 \def \bPhi {\bfsym \Phi}
 \def\bpsi{\bfsym \psi}             

\def\1{\bfsym{1}}	

\def\newpage{\vfill\eject}

\def\today{\ifcase\month\or
  January\or February\or March\or April\or May\or June\or
  July\or August\or September\or October\or November\or December\fi
  \space\number\day, \number\year}

\newdimen\biblioindent    \biblioindent=30pt

 at 8truept

\newcommand{\beqn}{\begin{eqnarray}}
  \newcommand{\eeqn}{\end{eqnarray}}
\newcommand{\beqnn}{\begin{eqnarray*}}
  \newcommand{\eeqnn}{\end{eqnarray*}}

\allowdisplaybreaks
\usepackage{verbatim}
\def\tilde{\widetilde}

\def\[{\left [}  \def\]{\right ]} \def\({\left (}  \def\){\right )}
 
\def\hat{\widehat}

\theoremstyle{definition}

\def \diag {\mathrm{diag}} 
\def \det {\mathrm{det}}

\begin{document}
	\bibliographystyle{ecta}
	
	\title{Large Volatility Matrix Prediction using Tensor Factor Structure}

	\date{\today}

	\author{
		Sung Hoon Choi\thanks{%
			Department of Economics, University of Connecticut, Storrs, CT 06269, USA. 
			E-mail: \texttt{sung\_hoon.choi@uconn.edu}.} \\ 
		\and Donggyu Kim\thanks{Department of Economics, University of California, Riverside, CA 92521, USA.
			Email: \texttt{donggyu.kim@ucr.edu}.}\\ 
		}
	\maketitle
	\pagenumbering{arabic}	
\begin{abstract}
		\onehalfspacing
        Several approaches for predicting large volatility matrices have been developed based on high-dimensional factor-based Itô processes. 
        These methods often impose restrictions to reduce the model complexity, such as constant eigenvectors or factor loadings over time.
        However, several studies indicate that eigenvector processes are also time-varying. 
        To address this feature, this paper generalizes the factor structure by representing the integrated volatility matrix process as a cubic (order-3 tensor) form, which is decomposed into low-rank tensor and idiosyncratic tensor components.
        To predict conditional expected large volatility matrices, we propose the Projected Tensor Principal Orthogonal componEnt Thresholding (PT-POET) procedure and establish its asymptotic properties. 
        The advantages of PT-POET are validated through a simulation study and demonstrated in an application to minimum variance portfolio allocation using high-frequency trading data.\\
        
        %In this paper, we develop a novel method for predicting future large volatility matrices based on high-dimensional factor-based Itô processes. 
        %Several studies have proposed volatility matrix prediction methods using parametric models to account for volatility dynamics. 
        %However, these methods often impose restrictions, such as constant eigenvectors over time. 
        %To generalize the factor structure, we construct a cubic (order-3 tensor) form of an integrated volatility matrix process, which can be decomposed into low-rank tensor and idiosyncratic tensor components. 
        %To predict conditional expected large volatility matrices, we introduce the Projected Tensor Principal Orthogonal componEnt Thresholding (PT-POET) procedure and establish its asymptotic properties. 
        %Finally, the advantages of PT-POET are also verified by a simulation study and illustrated by applying minimum variance portfolio allocation using high-frequency trading data.\\
        		
					
		%\noindent \textbf{JEL Classification:} 
\end{abstract}

\newpage 

\doublespacing

\section{Introduction}
The study of volatility using high-frequency financial data is a pivotal area of research in financial econometrics and statistics. 
Understanding the dynamics of asset return volatility is critical for practical applications, including hedging, option pricing, risk management, and portfolio optimization. 
The growing availability of high-frequency financial data has spurred the development of numerous effective non-parametric methods for estimating integrated volatility. 
Notable examples include two-time scale realized volatility (TSRV) \citep{zhang2005tale}, multi-scale realized volatility (MSRV) \citep{zhang2006efficient, zhang2011estimating}, pre-averaging realized volatility (PRV) \citep{christensen2010pre, jacod2009microstructure}, wavelet realized volatility (WRV) \citep{fan2007multi}, kernel realized volatility (KRV) \citep{barndorff2008designing, barndorff2011multivariate}, quasi-maximum likelihood estimator (QMLE) \citep{ait2010high, xiu2010quasi}, local method of moments \citep{bibinger2014estimating}, and robust pre-averaging realized volatility \citep{fan2018robust, shin2023adaptive}.

The use of high-frequency data has greatly enhanced our understanding of market dynamics at lower (e.g., daily) frequencies.
To capture these dynamics, various conditional volatility models based on realized volatility have been developed.
Examples include realized volatility-based modeling approaches \citep{andersen2003modeling}, heterogeneous autoregressive (HAR) models \citep{corsi2009simple}, high-frequency-based volatility (HEAVY) models \citep{shephard2010realising}, realized GARCH models \citep{hansen2012realized}, and unified GARCH-Itô models \citep{kim2016unified, song2021volatility}.
Their empirical studies typically focus on volatility dynamics, given the high-frequency information for a finite number of assets.
However, in practice, we often need to manage large portfolios, which causes the overparameterization issue due to the excessive number of parameters compared to the sample size.
To address this issue, approximate factor model structures are commonly imposed on large volatility matrices \citep{fan2013large}.
In particular, high-dimensional factor-based Itô processes are frequently used under sparsity assumptions on the idiosyncratic volatility \citep{ait2017using, fan2016incorporating, fan2018robust, kim2018large}. 
Recently, \cite{kim2019factor} proposed the factor GARCH-Itô model, and \cite{shin2021factor} developed the factor and idiosyncratic VAR-Itô model, both based on high-dimensional factor-based Itô processes. 
These models assume that the eigenvalue sequence of latent factor volatility matrices follows either a unified GARCH-Itô structure \citep{kim2016unified} or a VAR structure, which allows the dynamics of volatility to be explained by factors. 
They restrict that the eigenvectors of the latent factor volatility matrices are time-invariant.
However, several empirical studies have shown that eigenvectors are time-varying \citep{kong2018testing, kong2023discrepancy, su2017time}.
Thus, accommodating the eigenvector dynamics is essential to better account for the dynamics of large volatility matrices.

This paper proposes a novel prediction approach for a future large volatility matrix.
Specifically, we represent a large volatility matrix process in a cubic (order-3 tensor) form by stacking large integrated volatility matrices over time to capture interday time series dynamics.
To address the high-dimensionality problem, we impose a low-rank factor structure \citep{chen2023statistical, de2000multilinear, kolda2009tensor} and sparse idiosyncratic structure on the tensor.
The low-rank tensor component represents a conditional expected factor volatility tensor, which follows a semiparametric factor structure \citep{chen2024}, and we apply the Projected-PCA \citep{fan2016projected} procedure to estimate the time series loading matrix.
%To implement the proposed integrated volatility matrix model, we apply the Projected-PCA method suggested by \cite{fan2016projected} to estimate the time series loading matrix, which reflects the daily volatility dynamics.
%For example, we project the loading matrix onto a linear space spanned by past realized volatility estimators, which allows us to study the daily volatility dynamics and to predict the one-day-ahead large volatility matrix using current observed volatility information.
To account for the sparse idiosyncratic volatility structure, after removing the projected factor volatility component, we adopt the principal orthogonal component thresholding (POET) \citep{fan2013large} procedure. 
This method is called the Projected Tensor POET (PT-POET) procedure.
We then derive convergence rates for the projected integrated volatility matrix estimator and the predicted large volatility matrix using the PT-POET approach.
In an empirical study, we demonstrate that the proposed PT-POET estimator performs well in out-of-sample predictions for the one-day-ahead large volatility matrix and in a minimum variance portfolio allocation analysis using high-frequency trading data.

The rest of the paper is organized as follows. 
Section \ref{section2} establishes the model and introduces the PT-POET method to predict the conditional expected large volatility matrix. 
Section \ref{asymp} develops an asymptotic analysis of the PT-POET estimator. 
The merits of the proposed method are demonstrated by a simulation study in Section \ref{simulation} and by real data application on predicting the one-day-ahead volatility matrix and portfolio allocation in Section \ref{empiric}. 
Section \ref{conclusion} concludes the study. 
All proofs are presented in the online supplement file.

\section{Model Setup and Estimation Procedure} \label{section2}
Throughout this paper, we denote by $\|\bA\|_{F}$, $\|\bA\|_{2}$ (or $\|\bA\|$ for short), $\|\bA\|_{1}$, $\|\bA\|_{\infty}$, and $\|\bA\|_{\max}$ the Frobenius norm,  operator norm, $l_{1}$-norm, $l_{\infty}$-norm, and elementwise norm, which are defined, respectively, as $\|\bA\|_{F} = \tr^{1/2}(\bA'\bA)$, $\|\bA\| = \lambda_{\max}^{1/2}(\bA'\bA)$, $\|\bA\|_{1} = \max_{j}\sum_{i}|a_{ij}|$, $\|\bA\|_{\infty} = \max_{i}\sum_{j}|a_{ij}|$, and $\|\bA\|_{\max} = \max_{i,j}|a_{ij}|$. 
We use $\lambda_{\min}(\bA)$ and $\lambda_{\max}(\bA)$ to denote the minimum and maximum eigenvalues of a matrix $\bA$.
We denote by $\sigma_{i}(\bA)$ the $i$-th largest singular value of $\bA$.
When $\ba$ is a vector, the maximum norm is denoted as $\|\ba\|_{\infty}=\max_{i}|a_{i}|$, and both $\|\ba\|$ and $\|\ba\|_{F}$ are equal to the Euclidean norm. 

For a tensor $\cA \in \mathbb{R}^{I_1 \times I_2 \times I_3}$, we define its mode-1 matricization as a $I_1 \times I_2 I_3$ matrix $\mathcal{M}_1(\cA)$ such that 
$[\mathcal{M}_1(\cA)]_{i_1, i_2 + (i_3 - 1)I_2} = a_{i_1 i_2 i_3}$ for all $i_1 \in [I_1], i_2 \in [I_2], i_3 \in [I_3]$. 
%the mode-1 slices of $\cA$ are matrices $\bA_{i_1::} \in \mathbb{R}^{I_2 \times I_3}$ for any $i_1 \in [I_1]$, and the mode-1 fibers of $\cA$ are vectors $\ba_{:i_2i_3} \in \mathbb{R}^{I_1}$ for any $i_2 \in [I_2]$ and $i_3 \in [I_3]$. 
%In other words, matrix $\mathcal{M}_1(\mathcal{S})$ consists of all mode-1 fibers of $\cA$ as columns. 
For a tensor $\mathcal{F} \in \mathbb{R}^{R_1 \times R_2 \times R_3}$ and a matrix $\mathbf{A}_1 \in \mathbb{R}^{I_1 \times R_1}$, the mode-1 product is a mapping defined as 
$\times_1 : \mathbb{R}^{R_1 \times R_2 \times R_3} \times \mathbb{R}^{I_1 \times R_1} \mapsto \mathbb{R}^{I_1 \times R_2 \times R_3}$ as 
$\mathcal{F} \times_1 \mathbf{A}_1 = [\sum_{r_1 = 1}^{R_1} a_{i_1 r_1} f_{r_1 r_2 r_3}]_{i_1 \in [I_1], r_2 \in [R_2], r_3 \in [R_3]}.$
Similarly, we can define mode matricization and mode product for mode-2 and mode-3, respectively.

\subsection{A Model Setup} \label{model}

Denote by $\bX^{l}(t) = (X^{l}_{1}(t),\dots,X^{l}_{p}(t))^{\top}$ the vector of true log-prices of $p$ assets at the $l$-th day and intraday time $t \in [0,1]$.
To account for cross-sectional dependence, we consider the following factor-based jump diffusion model: for each $l = 1,\dots, D$,
\begin{equation} \label{diffusion-def}
d\bX^{l}(t) = \bmu^{l}(t)dt + \bB^{l}(t)d\bff^{l}(t) + d \bu^{l}(t) + \bJ^{l}(t)d \bLambda^{l}(t),   
\end{equation}
where $\bmu^{l}(t) \in \RR^{p}$ is a drift vector, $\bB^{l}(t) \in \RR^{p\times r}$ is an unknown factor loading matrix, $\bff^{l}(t) \in \RR^{r}$ is a latent factor process, and $\bu^{l}(t)$ is an idiosyncratic process.
In addition, for the jump part, $\bJ^{l}(t) = (J_{1}(t),\dots,J_{p}(t))^{\top}$ is a jump size vector, and $\bLambda^{l}(t) = (\Lambda^{l}_{1}(t),\dots,\Lambda^{l}_{p}(t))^{\top}$ is a $p$-dimensional Poisson process with an intensity vector $\bI(t) = (I_{1}(t),\dots,I_{p}(t))^{\top}$.
Assume that the latent factor and idiosyncratic processes $\bff^{l}(t)$ and $\bu^{l}(t)$ follow the continuous-time diffusion models as follows: for each $l = 1,\dots,D$,
\begin{equation}
    d\bff^{l}(t) = \bvartheta^{l\top}(t)d\bW^{l}(t) \text{ and } d\bu^{l}(t) = \bsigma^{l\top}(t)d\bW^{l*}(t),
\end{equation}
where $\bvartheta^{l}(t)$ is an $r_{1} \times r_{1}$ matrix, $\bsigma^{l}(t)$ is a $p \times p$ matrix, $\bW^{l}(t)$ and $\bW^{l*}(t)$ are $r_{1}$-dimensional and $p$-dimensional independent Brownian motions, respectively.

Stochastic processes $\bmu^{l}(t), \bX^{l}(t), \bff^{l}(t), \bu^{l}(t), \bB^{l}(t), \bsigma^{l}(t)$ and $\bvartheta^{l}(t)$ are defined on a filtered probability space $(\Omega, \cI, \cI_t, t \in [0, \infty)\}, P)$ with filtration $\cI_{t}$ satisfying the usual conditions. 
We note that the time unit in our applications is the day, and high-frequency intra-daily asset data is observed.
The instantaneous volatility of $\bX^{l}(t)$ is
\begin{equation}
\bgamma^{l}(t) = (\gamma^{l}_{ij}(t))_{i,j=1,\dots,p} = \bB^{l}(t) \bvartheta^{l\top}(t) \bvartheta^{l}(t) \bB(t)^{l\top} + \bsigma^{l\top}(t) \bsigma^{l}(t),
\end{equation}
and the integrated volatility for the $l$-th day is
\begin{equation} \label{integrated_vol}
\bGamma_l = (\Gamma_{l,ij})_{i,j=1,\dots,p} = \int_{l-1}^l \bgamma^{l}(t) dt = \bPsi_l + \bSigma_l,    
\end{equation}
where $\bPsi_l = \int_{l-1}^l\bB^{l}(t) \bvartheta^{l\top}(t) \bvartheta^{l}(t) \bB(t)^{l\top} dt$ and $\bSigma_l = \int_{l-1}^l\bsigma^{l\top}(t) \bsigma^{l}(t) dt$.
For each $l =1,\dots,D$,  the integrated volatility matrix $\bGamma_{l}$ has the low-rank plus sparse structure \citep{ait2017using, kim2019factor, fan2008high, fan2013large}.
Specifically, the factor volatility matrices $\bPsi_{l}$ has the finite rank $r_{1}$, and the idiosyncratic volatility matrices $\bSigma_{l} = (\Sigma_{l,ij})_{i,j =1,\dots,p}$ is sparse as follows:
\begin{align}
    \max_{1\leq l \leq D}\max_{1\leq i\leq p}\sum_{1\leq j\leq p}|\Sigma_{l,ij}|^{\eta}(\Sigma_{l,ii}\Sigma_{l,jj})^{(1-\eta)/2} = O(s_{p}),
\end{align}
for some $\eta \in [0,1)$, the sparsity measure $s_{p}$ diverges slowly, such as $\log p$.
We note that when $\eta=0$, $s_{p}$ measures the maximum number of non-zero elements in each row of $\bSigma_{l}$.

We can write the integrated volatility matrix process of the model \eqref{integrated_vol} in a cubic (order-3 tensor) form as follows: 
\begin{equation} \label{tensor model}
    \cY = (\bGamma_{l})_{l=1,\dots, D}  = \cF \times_{1} \bQ \times_{2} \bQ \times_{3} \bV + \cE := \cS + \cE,
\end{equation}
where $\cY \in \RR^{p\times p \times D}$, $\cF$ is the $r_{1}\times r_{1} \times r_{2}$ latent tensor factor, $\bQ=(q_{i,k_{1}})_{i=1,\dots,p, k_{1} = 1,\dots, r_{1}}$ is the $p \times r_{1}$ loading matrix corresponds to the integrated volatility matrix, and $\bV = (v_{l,k_{2}})_{l=1,\dots,D, k_{2} = 1,\dots, r_{2}}$ is the $D \times r_{2}$ time series loading matrix corresponds to daily volatility dynamics.
$\cS$ is the factor volatility tensor, which has a Tucker decomposition such that $\bQ$ and $\bV$ are orthonormal matrices of the left singular vectors of $\cM_{1}(\cS)$ and $\cM_{3}(\cS)$, respectively.
We refer to $\cE = (\bSigma_{l})_{l=1,\dots,D}$ as the idiosyncratic volatility tensor.
We note that the time series loading matrix explains daily volatility dynamics, which are often driven by past realized volatilities \citep{corsi2009simple, hansen2012realized,  kim2019factor, kim2016unified, shephard2010realising, song2021volatility}.
For instance, following the HAR model \citep{corsi2009simple}, we can consider $v_{l,1} = b_0 + b_1 \zeta_{l-1} + b_2  \frac{1}{5}\sum_{j=1}^5 \zeta_{l-j} +  b_3 \frac{1}{21}\sum_{j=1}^{21} \zeta_{l-j}$, where $\zeta_l$ is the $l$-th day realized largest eigenvalue. 
This feature motivates the representation of the cubic structure in the model \eqref{tensor model}, and we propose a generalized model hereafter.

In this paper, our target is to predict the one-day-ahead integrated volatility matrix.
In general, we assume that $v_{l, k_{2}}$ is $\cI_{l-1}$-adapted and the idiosyncratic volatility matrices $\bSigma_{l}$ are martingale processes such that $E(\bSigma_{D+1} | \cI_{D}) = \bSigma_{D}$ a.s. 
Thus, given the current information $\cI_{D}$, we can predict the integrated volatility matrix as follows.  
The conditional expected large volatility matrix $\bGamma_{D+1}$ is
\begin{equation}
    E\left[\bGamma_{D+1} | \cI_{D}\right] =  \cF \times_{1} \bQ \times_{2} \bQ \times_{3} \bv_{D+1} + \bSigma_{D} \text{ a.s.,}
\end{equation}
where $\bv_{D+1} := (v_{D+1,1}, \dots v_{D+1,k_{2}})$.
Based on \eqref{integrated_vol} and \eqref{tensor model}, we consider the following nonparametric structure on the time series loading components, which is modeled as additive via sieve approximations \citep{fan2016projected}:  for each $k \leq r_{2}$ and $l \leq D$,
  \begin{align}\label{nonparametric function}
  	v_{l,k} := g_{k}(\bx_{l}) = \phi(\bx_{l})'\ba_{k} + R_{k}(\bx_{l}),
  \end{align}
  where $\bx_{l} = (x_{i1},\dots,x_{ld})$ is observable covariates that explain the time series loading vectors, $\phi(\bx_{l})$ is a $(Jd)\times 1$ vector of basis functions, $\ba_{k}$ is a $(Jd)\times 1$ vector of sieve coefficients, and $R_{k}(\bx_{l})$ is the approximation error term.
In this context, for example, $\bx_{l}$ can be the past eigenvalues in the VAR model \citep{shin2021factor} or realized largest eigenvalues of yesterday, last week, and last month in the HAR model \citep{corsi2009simple}.
We note that \eqref{nonparametric function} can be written as $g_{k}(\bx_{l}) = \sum_{m=1}^{d}g_{km}(x_{lm})$, where $g_{km}(x_{lm}) = \sum_{j=1}^{J}b_{j,km}\phi_{j}(x_{jm})+R_{km}(x_{jm})$.
Hence, the additive component of $g_{k}$ can be estimated by the sieve method.
We assume that $d$ is fixed, and the number of sieve terms $J$ grows very slowly as $D \rightarrow \infty$.
  In a matrix form, we can write
\begin{align}
    \bV := \bG(\bX) = \bPhi(\bX)\bA + \bR(\bX),
\end{align}
where the $D \times (Jd)$ matrix $\bPhi(\bX) = (\phi(\bx_{1}), \dots, \phi(\bx_{D}))'$, the $(Jd)\times r_{2}$ matrix $\bA = (\ba_{1},\dots, \ba_{r_{2}})$, and $\bR(\bX) = (R_{k}(\bx_{l}))_{D\times r_{2}}$.

The true underlined log-price $\bX^{l}(t)$ in \eqref{diffusion-def} cannot be observed because of the imperfections of the trading mechanisms.
Hence, we assume that the high-frequency intraday observations are contaminated by microstructure noises:
\begin{align}
    Y_{i}(t_{l,k}) = X_{i}(t_{l,k}) + u_{i}(t_{l,k}), \quad i = 1,\dots, p, l = 1,\dots, D, k = 0,\dots, m,
\end{align}
where $l-1 = t_{l,0}<\cdots < t_{l,m} = l$,  and the microstructural noises are random variables with a mean of zero.
%For simplicity,  we assume that the observed time points are synchronized and equally spaced such that $t_j-t_{j-1} = m^{-1}$ for $i = 1,\dots, D$ and $j = 2,\dots, m$.

Several studies have developed a non-parametric integrated volatility matrix that is robust to jumps and dependent structures of the microstructure noise \citep{ait2016increased, barndorff2011subsampling, bibinger2015econometrics, jacod2009microstructure, koike2016quadratic, li2022remedi, shin2023adaptive}.
We can employ any well-performing realized volatility matrix estimator that satisfies Assumption \ref{basic_assumptions} (v).
In the numerical study, we utilize the jump-adjusted pre-averaging realized volatility matrix (PRVM) estimator \citep{ait2016increased,christensen2010pre, jacod2009microstructure} as described in \eqref{PRVM}.

\subsection{Projected Tensor POET} \label{estimation procedure}

To utilize the semiparametric structure outlined in Section \ref{model}, it is necessary to project the time series loading vectors onto linear spaces spanned by the corresponding covariates.
For this purpose, we apply the Projected-PCA \citep{fan2016projected} procedure to the time series loading matrix with the well-performing integrated volatility tensor estimator. 
The specific procedure is as follows:
\begin{enumerate}
  	\item For each $l \leq D$, we estimate the integrated volatility matrix, $\bGamma_{l}$, using a non-parametric estimation method with high-frequency log-price observations and denote them by $\hat{\bGamma}_{l}$. 
    Let $\Bar{\bPsi}_{l} = \sum_{i=1}^{r_1}\bar{\delta}_{l,i}\bar{\xi}_{l,i}\bar{\xi}_{l,i}'$, where $\{\bar{\delta}_{l,i},\bar{\xi}_{l,i}\}_{i=1}^{p}$ are the eigenvalues and eigenvectors of $\hat{\bGamma}_{l}$ in decreasing order.
    Let $\hat{\cY} = (\hat{\bGamma}_{l})_{l=1,\dots,D}$ and $\bar{\cS} = (\bar{\bPsi}_{l})_{l=1,\dots,D}$, which are $p \times p \times D$ tensors.
   \item  Define the $D\times D$ projection matrix as $\bP = \bPhi(\bX)(\bPhi(\bX)'\bPhi(\bX))^{-1}\bPhi(\bX)'$.
   The columns $\hat{\bQ}$ are defined as the $r_{1}$ leading left singular vectors of $\cM_{1}(\tilde{\cS})$, where $\tilde{\cS} = \bar{\cS} \times_{3} \bP$.
   \item The columns $\hat{\bV} :=\hat{\bG}(\bX)$ are defined as the $r_{2}$ leading left singular vectors of the $\cM_{3}(\tilde{\cS})$.
   Then, we can estimate $\bA$ by
  	\begin{align*}
  	\hat{\bA} = (\hat{\ba}_{1},\dots,\hat{\ba}_{r_{2}}) = (\bPhi(\bX)'\bPhi(\bX))^{-1}\bPhi(\bX)'\hat{\bG}(\bX).
  	\end{align*}
    Given any $\bx \in \calX$, we estimate $g_{k}(\cdot)$ by
   $$
   \hat{g}_{k}(\bx) = \phi(\bx)'\hat{\ba}_{k} \qquad \text{for} \quad k = 1,\dots, r_{2},
   $$
   where $\calX$ denotes the support of $\bx_{i}$.
   \item We estimate the latent tensor factor by
   \begin{equation*}
       \hat{\cF} = \tilde{\cS} \times_{1} \hat{\bQ}^{\top} \times_{2} \hat{\bQ}^{\top} \times_{3} \hat{\bG}(\bX)^{\top}.
   \end{equation*}
   Then, we compute the principal orthogonal complement tensor as follows:
    \begin{equation*}
        \hat{\cE} = \hat{\cY} - \hat{\cS},
    \end{equation*}
    where the factor volatility tensor estimator $\hat{\cS} = \hat{\cF} \times_{1} \hat{\bQ} \times_{2} \hat{\bQ} \times_{3} \hat{\bG}(\bX)$.

    \item Define $\hat{\cE} := (\tilde{\bSigma}_{l})_{l=1,\dots,D}$. For each $l$, we apply the adaptive thresholding method to $\tilde{\bSigma}_{l} = (\tilde{\Sigma}_{l,ij})_{1\leq i,j \leq p}$:
    \begin{equation*}
			\hat{\bSigma}_{l} = (\hat{\Sigma}_{l,ij})_{1\leq i,j\leq p}, \;\;\;\;\; \hat{\Sigma}_{l,ij} = \left\{ \begin{array}{ll}
				\tilde{\Sigma}_{l,ij} \vee 0, & i=j \\ 
				s_{ij}(\tilde{\Sigma}_{l,ij})I(|\tilde{\Sigma}_{l,ij}| \geq \tau_{ij}), & i \neq j
			\end{array}\right., 
		\end{equation*}
		where an entry-dependent threshold $\tau_{ij} = \tau\sqrt{(\tilde{\Sigma}_{l,ii}\vee 0)(\tilde{\Sigma}_{l,jj}\vee 0)}$ and $s_{ij}(\cdot)$ is a generalized thresholding function (e.g., hard or soft thresholding; see  \citealt{cai2011adaptive, rothman2009generalized}).
    The thresholding constant $\tau$ will be determined in Theorem \ref{main_thm}.

    \item Finally, we predict the conditional expected volatility matrix $E\left[\bGamma_{D+1} | \cI_{D}\right]$ by
    \begin{equation*}
        \hat{\bGamma}_{D+1} = \hat{\cF} \times_{1} \hat{\bQ} \times_{2} \hat{\bQ} \times_{3} \hat{\bg}(\bx_{D+1}) + \hat{\bSigma}_{D+1},
    \end{equation*}
    where $\hat{\bg}(\bx_{D+1}) = (\hat{g}_{1}(\bx_{D+1}),\dots, \hat{g}_{r_{2}}(\bx_{D+1}))$ and $\hat{\bSigma}_{D+1} = \frac{1}{D}\sum_{l=1}^{D}\hat{\bSigma}_{l}$.
  \end{enumerate}

In summary, given the estimated integrated volatility tensor, we apply the Projected-PCA method \citep{fan2016projected} to estimate the unknown nonparametric function using observable covariates, such as a series of past realized volatilities. 
Based on the projected tensor data, we then use tensor singular value decomposition to estimate loading matrices as well as the latent tensor factor.
Then, we apply the thresholding method to the remaining residual component after removing the factor volatility tensor estimator.
Finally,  we predict the one-day-ahead realized volatility matrix by multiplying the estimated tensor factor and loading components, using the observable covariates $\bx_{D+1}$, such as the realized volatility information on the $D$th day.
We call this procedure the Projected Tensor Principal Orthogonal complEment Thresholding (PT-POET).
The PT-POET method can accurately predict the factor volatility matrix by incorporating daily volatility dynamics via the projection approach based on the tensor structure.
The numerical analyses in Sections \ref{simulation} and \ref{empiric} demonstrate that PT-POET performs well in predicting a one-day-ahead realized volatility matrix.

\subsection{Choice of Tuning Parameters} \label{tuning}

To implement PT-POET, we need to determine tuning parameters $r_{1}$, $r_{2}$, and $J$. 
Several studies suggested data-driven methods to consistently estimate the number of factors by finding the largest singular value gap or singular value ratio  \citep{ahn2013eigenvalue, bai2002determining, lam2012factor, onatski2010determining}.
In this context, the rank of each dimension can be determined based on the matricized tensor \citep{chen2024}.
Specifically, $r_{1}$ and $r_{2}$ can be estimated as follows: for each $s \in \{1,3\}$, $\hat{r}_{s} = \argmax_{k\leq r_{\max}} (\sigma_{k}(\cM_{s}(\hat{\cY}))-\sigma_{k+1}(\cM_{s}(\hat{\cY})))$ or $\hat{r}_{s} = \argmax_{k\leq r_{\max}} \frac{\sigma_{k}(\cM_{s}(\hat{\cY}))}{\sigma_{k+1}(\cM_{s}(\hat{\cY}))}$ for a predetermined maximum number of factors $r_{\max}$.
For the numerical studies in Sections \ref{simulation} and \ref{empiric}, we employed $\hat{r}_{1}=3$ and $\hat{r}_{2} = 1$ using the eigenvalue ratio method proposed by \cite{ahn2013eigenvalue} and the rank choice method considered in \cite{ait2017using}.

Practitioners can flexibly choose the number of sieve terms, $J$, and the basis functions based on conjectures about the form of the nonparametric function \citep{chen2024, fan2016projected}.
In this paper, since the daily integrated volatility dynamic is a linear function of past realized volatilities, we employed an additive polynomial basis with a sieve dimension of $J=2$ for the numerical studies in Sections \ref{simulation} and \ref{empiric}.

\section{Asymptotic Properties} \label{asymp}
This section establishes the asymptotic properties of the proposed PT-POET estimator.
To do this, we impose the following technical assumptions.

    \begin{assum} \label{basic_assumptions} ~
      \begin{itemize}
            \item[(i)] Let $D_{\delta} = \min\{\tilde{\delta}_{l,i} - \tilde{\delta}_{l,i+1}, i = 1,\dots, r_{1}, l = 1,\dots, D\}$, where $\tilde{\delta}_{l,i}$ is the leading eigenvalues of $\bPsi_{l}$ and $\tilde{\delta}_{l,i+1} = 0$.
            For some positive constant $C_{1}$, $D_{\delta} \geq C_{1}p$ a.s.
            \item[(ii)] For some fixed constant $C_{2}$, we have
            $$
            \frac{p}{r_{1}}\max_{1\leq i\leq p}\sum_{j=1}^{r_{1}}\iota_{l,ij}^{2} \leq C_{2},
            $$            
            where $\tilde{\xi}_{l,j} = (\iota_{l,1j},\dots, \iota_{l,pj})^{\top}$ is the $j$th eigenvector of $\bPsi_{l}$ for each $l \leq D$.
            \item[(iii)] For $k_{1} \leq r_{1}$ and $k_{2} \leq r_{2}$, the eigengap satisfies 
            \begin{align*}
                &|\sigma_{k_{1}}(\cM_{1}(\cF)) -\sigma_{k_{1}+1}(\cM_{1}(\cF))| = O_{P}(p\sqrt{D}),\\
                &|\sigma_{k_{2}}(\cM_{3}(\cF)) -\sigma_{k_{2}+1}(\cM_{3}(\cF))| = O_{P}(p\sqrt{D}).
            \end{align*}
            In addition, $\cM_{m}(\cF)\cM_{m}^{\top}(\cF)$ is a diagonal matrix with non-zero decreasing singular values for $m=1,2,3$.
            \item[(iv)] There exist  constants $c_{1},c_{2}>0$ such that $\lambda_{\min}(\bSigma_{l}) > c_{1}$ and $\|\bSigma_{l}\|_{1} \leq c_{2}s_{p}$ for each $l \leq D$.
            \item[(v)] The estimated nonparametric estimator $\hat{\bGamma}_{l}$ satisfies
              $$
              \Pr\Bigg\{\max_{1\leq l\leq D}\| \hat{\bGamma}_{l}-\bGamma_{l}\|_{\max} \geq C\sqrt{\frac{\log(pD \vee m)}{m^{1/2}}}\Bigg\} \leq p^{-1},
              $$
              where $m$ is the number of observations of the process $\bX$ each day.
        \end{itemize}
  \end{assum}
  \begin{remark}
    Assumption \ref{basic_assumptions}(i) and (ii) are the pervasive condition and incoherence condition for daily integrated volatility matrices, respectively.
    These assumptions are often imposed when analyzing the approximate factor models or the low-rank matrix inference \citep{ait2017using, bai2003inferential, fan2013large, fan2018eigenvector, kim2019factor, stock2002forecasting}.
    Assumption \ref{basic_assumptions}(ii) is the eigengap assumption for the tensor factor $\cF$, which is essential for analyzing low-rank matrices \citep{candes2010matrix, cho2017asymptotic, fan2018large}.
    Since we have a $p\times p \times D$ realized volatility tensor, the pervasive condition implies that the eigenvalue for the low-rank component has $p\sqrt{D}$ order based on the matricized tensor for each mode.
    To analyze large tensor inferences, we impose the element-wise convergence condition (Assumption \ref{basic_assumptions}(v)).
    This condition can be satisfied under the bounded instantaneous volatility condition \citep{tao2013fast}. 
    Furthermore, under the locally bounded condition of the instantaneous volatility process with heavy-tailed observations, we can obtain the element-wise convergence condition \citep{fan2018robust, shin2021factor}.
  \end{remark}

    \begin{assum}\label{assum_phi} ~
        \begin{itemize}
              \item[(i)] There are $c_{\min}$ and $c_{\max}>0$ so that, with the probability approaching one, as $D \rightarrow \infty$,
              $$
              c_{\min} < \lambda_{\min}(D^{-1}\bPhi(\bX)'\bPhi(\bX))< \lambda_{\max}(D^{-1}\bPhi(\bX)'\bPhi(\bX))<c_{\max}.
              $$
              \item[(ii)] $\max_{j\leq J,l\leq D, d'\leq d} E\phi_{j}(x_{ld'})^{2} < \infty$, and $\max_{k\leq r_{2},l\leq D}Eg_{k}(\bx_{l})^{2} < \infty$.
        \end{itemize}
    \end{assum}
Assumption \ref{assum_phi} pertains to the basis functions. Intuitively, the strong law of large numbers implies Assumption \ref{assum_phi}(i), which can be satisfied by normalizing commonly used basis functions such as B-splines, polynomial series, or Fourier bases.

  \begin{assum} \label{assum_sieve}
      For all $d'\leq d, k \leq r_{2}$, 
      \begin{itemize}
          \item[(i)] the functions $g_{kd'}(\cdot)$ belong to a H\"older class $\mathcal{G}$ defined by, for some $L>0$,
		\begin{align*}
			&\mathcal{G} = \{g : |g^{(c)}(s)-g^{(c)}(t)| \leq L|s-t|^{\alpha}\};
        \end{align*}
          \item[(ii)] the sieve coefficients $\{a_{j,kd'}\}_{j\leq J}$ satisfy for $\kappa = 2(c+\alpha) \geq 4$, as $J \rightarrow \infty$, 
		\begin{equation*}
			\sup_{x \in \mathcal{X}_{d'}}\left|g_{kd'}(x)-\sum_{j=1}^{J}a_{j,kd'}\phi_{s}(x)\right|^{2} = O(J^{-\kappa}/D),
		\end{equation*}
		where $\mathcal{X}_{d}$ is the support of the $d'$th element of $\bx_{i}$;
          \item[(iii)] $\max_{k\leq r_{2},j\leq J,d'\leq d}a_{j,kd'}^2 = O(1/D)$.
        \end{itemize}
  \end{assum}

Assumption \ref{assum_sieve} is related to the accuracy of the sieve approximation and can be satisfied using a common basis such as a polynomial basis or B-splines \citep{chen2007large}.

We obtain the following elementwise convergence rates of the projected factor volatility matrix and sparse volatility matrix estimators.
    \begin{pro} \label{element wise norm}
    Suppose that Assumptions \ref{basic_assumptions}--\ref{assum_sieve} hold and $J = o(\sqrt{D})$. 
    Let $\omega_{m} =  m^{-\frac{1}{4}}\sqrt{J \log(pD \vee m)} +s_{p}\sqrt{J}/p+ J^{\frac{1-\kappa}{2}}$ and $\tau \asymp \omega_{m}$.
    As $m,p,D,J \rightarrow \infty$, we have
        \begin{align}
            &\max_{1\leq l \leq D}\|\hat{\bPsi}_{l} - \bPsi_{l}\|_{\max} = O_{P}(\omega_{m}),\label{factor_elementwise} \\
            &\max_{1\leq l \leq D}\|\hat{\bSigma}_{l} -\bSigma_{l}\|_{\max} = O_{P}(\omega_{m}), \quad \max_{1\leq l \leq D}\|\hat{\bSigma}_{l} -\bSigma_{l}\| = O_{P}(s_{p}\omega_{m}^{1-\eta}).\label{idio_norms}
        \end{align}    
    \end{pro}
\begin{remark} \label{remark1}
Proposition \ref{element wise norm} represents that the projected factor and idiosyncratic volatility matrix estimators have the convergence rate $m^{-\frac{1}{4}}\sqrt{J} + \frac{\sqrt{J}}{p} + J^{\frac{1}{2}-\frac{\kappa}{2}}$ up to the log order and the sparsity level under the max norm.
Specifically, $J$ relates to the sieve approximation, which accounts for the cost of approximating the unknown nonparametric function $g_k(\cdot)$.
If $g_k(\cdot)$ is known, the rate of convergence is $m^{-\frac{1}{4}} + p^{-\frac{1}{2}}$ up to the log order and the sparsity level, which is similar to that of \cite{kim2019factor}.
The term $m^{-1/4}$ reflects the cost of estimating the unobserved integrated volatility matrix using high-frequency data, which is the optimal rate for the instantaneous volatility estimator in the presence of microstructure noise.
The term  $p^{-1/2}$ represents the cost of identifying the latent factor volatility.
\end{remark}

The following theorem provides the convergence rate of the future volatility matrix estimator using the PT-POET method.
    \begin{thm} \label{main_thm}
    Suppose that Assumptions \ref{basic_assumptions}--\ref{assum_sieve} hold and $J = o(\sqrt{D})$.
    Let $\delta_{m,p,D} = (m^{-\frac{1}{4}}J\sqrt{\log(pD \vee m)} + s_{p}J/p + J^{1-\frac{\kappa}{2}} )\max_{j\leq J}\sup_{x}|\phi_{j}(x)|$, $\omega_{m} = m^{-\frac{1}{4}}\sqrt{J \log(pD \vee m)} +s_{p}\sqrt{J}/p+ J^{\frac{1-\kappa}{2}}$ and $\tau \asymp \omega_{m}$.
    As $m,p,D,J \rightarrow \infty$, we have
        \begin{align}
            &\|\hat{\bGamma}_{D+1} -E(\bGamma_{D+1}|\cI_{D})\|_{\max} = O_{P}(\delta_{m,p,D}),\label{thm_maxnorm}\\
            &\|\hat{\bGamma}_{D+1} -E(\bGamma_{D+1}|\cI_{D})\|_{\bGamma^{*}} = O_{P}\left(\delta_{m,p,D}+\frac{s_{p}}{\sqrt{p}} + \sqrt{p}\delta_{m,p,D}^{2} + s_{p}\omega_{m}^{1-\eta}\right),\label{thm_relativefro}
        \end{align}    
        where the relative Frobenius norm $\|\bA\|_{\bGamma^{*}}^{2} = p^{-1}\|\bGamma^{*-1/2}\bA\bGamma^{*-1/2}\|_{F}^{2}$ and $\bGamma^{*} = E(\bGamma_{D+1}|\cI_{D})$.
    \end{thm}
\begin{remark}
Theorem \ref{main_thm} shows that   PT-POET consistently predicts the conditional expected volatility matrix under both the max and relative Frobenius norms.
As discussed in Remark \ref{remark1}, there are $m^{-1/4}$, $\sqrt{J}/P$, and $J$ terms.
We note that out-of-sample predictions with any covariates $\bx$ require additional costs such as $\sqrt{J}$ and the supremum term $\max_{j\leq J}\sup_{x}|\phi_j(x)|$, whereas the result in Proposition \ref{element wise norm} does not require these costs since it pertains to in-sample prediction.
Under the relative Frobenius norm, the additional term $s_{p}/\sqrt{p}$ arises from estimating the factor component, while the term $s_{p}\omega_{m}^{1-\eta}$ results from estimating the sparse idiosyncratic component.
\end{remark}

\section{Simulation Study} \label{simulation}
In this section, we conducted a simulation study to examine the finite sample performances of the proposed PT-POET method. 
We first generated the log-prices $\bX^{l}(t_{j})$ for $D+1$ days with frequency $1/m$ on each day as follows: for $l = 1,\dots,D+1$, $j = 0, \dots, m$, and $t_{j} = j/m$,
\begin{align*}
    &\bY^{l}(t_{j}) = \bX^{l}(t_{j}) + \be^{l}(t_{j}), \\
    &d\bX^{l}(t) = \bpsi^{l \top}d\bW(t) + \bsigma^{l \top}d\bW^{*}(t) + \bJ^{l}(t)d\bLambda^{l}(t),
\end{align*}
and we set the market microstructure noise as $\be^{l}(t_{j}) = (e^{l}_1(t_j), \dots, e^{l}_p(t_j))$, where $e^{l}_i(t_j)$ were from i.i.d. normal distribution with mean zero and standard deviation $0.01\sqrt{\Sigma_{ii}}$ for the $i$th asset, and  the initial values $\bX_0 = (0,\dots, 0)^{\top}$; $\bW(t)$ and $\bW^{*}(t)$ are $r$-dimensional and $p$-dimensional independent Brownian motions, respectively, $\bJ^{l}(t) = (J_{1}^{l}(t), \dots, J_{p}^{l}(t))^{\top}$ is the jump size vector, and $\bLambda^{l}(t) = (\Lambda^{l}_{1}(t),\dots,\Lambda^{l}_{p}(t))^{\top}$ is the Poission process with intensity $\bI(t) = (5,\dots,5)^{\top}$.
The jump size $J_{i}^{l}(t)$ was obtained from the independent Gaussian distribution with mean zero and standard deviation $0.05\sqrt{\int_{0}^1\gamma_{ii}(t)dt}$.
$\bpsi^{l}$ and $\bsigma^{l}$ are the Cholesky decompositions of $\bPsi_{l}$ and $\bSigma_{l}$, respectively, where the integrated volatility process components $\cS=(\bPsi_{l})_{l = 1,\dots, D+1}$ and $\cE = (\bSigma_{l})_{l= 1,\dots, D+1}$ were constructed as follows:
\begin{align*}
    \cY = \cS + \cE =  \cF \times_{1} \bQ \times_{2} \bQ \times_{3} \bV + \cE,
\end{align*}
where the $r_1\times r_1 \times r_2$ latent tensor factor $\cF$ and the $p \times r_1$ loading matrix $\bQ$ were generated from the first $r_1$ leading eigenvalues and eigenvectors of $AA'$, respectively, where each element of $A$ was taken from an i.i.d standard normal distribution; $\bV=(v_{1},\dots, v_{D+1})'$ was generated by $v_{l} = b_0 + b_{1}v_{l-1} + b_{2}\frac{1}{5}\sum_{s=1}^{5}v_{l-s} + b_{3}\frac{1}{21}\sum_{s=1}^{21}v_{l-s}  + \zeta_{l}$, where  $\zeta_{l} \sim \mathcal{N}(0,1)$.
The model parameters were set to be $b_{0} = 0.5, b_{1} =0.372, b_{2} = 0.343, b_{3} = 0.224$. 
We set ranks $r_1 = 3$ and $r_2 = 1$, $p=200$.

We obtained the sparse volatility component $\cE = (\bSigma_{l})_{l= 1,\dots, D+1}$ as follows:
let $\bd = \diag(d_{1}^{2},\dots,d_{p}^{2})$, where each $\{d_{i}\}$ was generated independently from Gamma $(\alpha, \beta)$ with $\alpha = \beta = 100$. 
We set $s = (s_{1},\dots, s_{p})'$ to be a sparse vector, where each $s_{i}$ was drawn from $\mathcal{N}(0,1)$ with probability $\frac{0.3}{\sqrt{p}\log{p}}$, and $s_{i} = 0$ otherwise. 
Then, we set a sparse error covariance matrix as $\bSigma = \bd + ss' - \diag\{s_{1}^{2},\dots,s_{p}^{2}\}$, and we let $\bSigma_{l} = \bSigma$ for each $l$.
In the simulation, we generated $\bSigma$ until it is positive definite.
We first generated high-frequency data with $m = \{250, 500, 2000\}$ for 200 consecutive days and used the subsampled log prices of the last $D$ days. 
We varied $D$ from 50 to 200, and the whole simulation procedure was repeated 500 times.

To estimate the integrated volatility matrices, we employed the pre-averaging realized volatility matrix (PRVM) estimator  $\hat{\bGamma}_{l} = (\hat{\Gamma}_{l,ij})_{1\leq i,j\leq p}$ 
\citep{ait2016increased,christensen2010pre, jacod2009microstructure} for each $l$-th day as follows: 
\begin{equation}\label{PRVM}
    \hat{\Gamma}_{l,ij} =  \frac{1}{\phi K} \sum_{k=1}^{m-K+1} \left\{ \Bar{Y}^{l}_i(t_{k}) \Bar{Y}^{l}_j(t_{k}) - \frac{1}{2} \hat{Y}^{l}_{i,j}(t_{k}) \right\} \mathbf{1} \left\{ \left| \Bar{Y}^{l}_i(t_{k}) \right| \leq u_{i,m} \right\} \mathbf{1} \left\{ \left| \widetilde{Y}^{l}_j(t_{k}) \right| \leq u_{j,m} \right\},
\end{equation}
where
\begin{align*}
&\Bar{Y}^{l}_i(t_{k}) = \sum_{s=1}^{K-1} g \left( \frac{s}{K} \right) (Y^{l}_i(t_{k+s}) - Y^{l}_i(t_{k+s-1})),\\
&\hat{Y}^{l}_{i,j}(t_{k}) = \sum_{s=1}^{K}\Bigg[\left\{ g\left( \frac{s}{K} \right) - g \left( \frac{s-1}{K} \right) \right\}^2 \\
& \qquad\qquad\qquad\qquad \times (Y^{l}_i(t_{k+s-1}) - Y^{l}_i(t_{k+s-2})) (Y^{l}_j(t_{k+s-1}) - Y^{l}_j(t_{k+s-2}))\Bigg], 
\end{align*}
where
$\phi = \int_0^1 g(t)^2 \, dt$, $\mathbf{1}\{\}$ is an indicator function, and $u_{i,m} = c_{i,u} m^{0.235}$ is a truncation parameter for some constant $c_{i,u}$.
We chose the bandwidth parameter $K = \lfloor m^{1/2} \rfloor$, weight function $g(x) = x \wedge (1-x)$, and $c_{i,u}$ as 7 times the sample standard deviation for the pre-averaged variables $m^{\frac{1}{4}}\Bar{Y}_{i}(t_{d,k})$.

With the aggregated volatility matrix estimates spanning $D$ days, $\hat{\cY} = (\hat{\bGamma}_{l})_{l=1,\dots, D}$, we examined the out-of-sample performance of predicting the one-day-ahead aggregated volatility matrix.
For comparison, the PRVM, POET, FIVAR, T-POET, and PT-POET methods were employed to predict $E(\bGamma_{D+1}|\cI_{D})$, given the past $D$ period observation.
In particular, for PT-POET, we utilized the past daily, weekly, and monthly averages of the top eigenvalues for $\bX$ with the additive polynomial basis and $J = 2$.
PRVM and POET represent the PRVM estimator (i.e., $\hat{\bSigma}_{D}$) and the POET estimator \citep{fan2013large} based on PRVM at the $D$th day, respectively.
We also employed the FIVAR method  \citep{shin2021factor} to model the factor part dynamics based on the PRVM estimator.
Specifically, we used the past $D$ days' observations to estimate the model parameters and the previous 21 days to estimate the time-invariant eigenvectors.
We do not model the idiosyncratic dynamics to compare the performance of the factor modeling. 
Details can be found in \cite{shin2021factor}.
T-POET represents the $D$th day matrix estimator based on the conventional estimation procedure for the tensor observation, $\hat{\cY}$, without using additional covariates.
The integrated volatility matrix has a low-rank plus sparse structure.
Hence, to estimate the idiosyncratic component of the POET, FIVAR, T-POET, and PT-POET estimators, we employed a soft thresholding scheme and used the thresholding level $\sqrt{2 \log p/m^{1/2}}$ as in \citet{kim2019factor}.
Their idiosyncratic volatility matrix estimators are the same.

Figure \ref{sim_plot} presents the average log Frobenius, max, Spectral, and relative Frobenius norm errors of the future volatility matrix estimators with $D = 50, 100, 150, 200$ and $m = 250, 500, 2000$.
We note that for each simulation, the target future volatility matrix is the same for each different pair of $D$ and $m$.
Figure \ref{sim_plot} shows that the PT-POET method performs best.
This is because PT-POET can accurately predict the future integrated volatility matrix by leveraging the HAR-based interday volatility dynamics and time-varying eigenvector in addition to eigenvalue.
In addition, the matrix errors of PT-POET tend to decrease as $D$ and $m$ increase.
This finding supports the theoretical results in Section \ref{asymp}.
%On the other hand, FIVAR underperforms PT-POET because it cannot incorporate the time-varying eigenvectors.

\begin{figure}[!ht]
	\includegraphics[width=\linewidth]{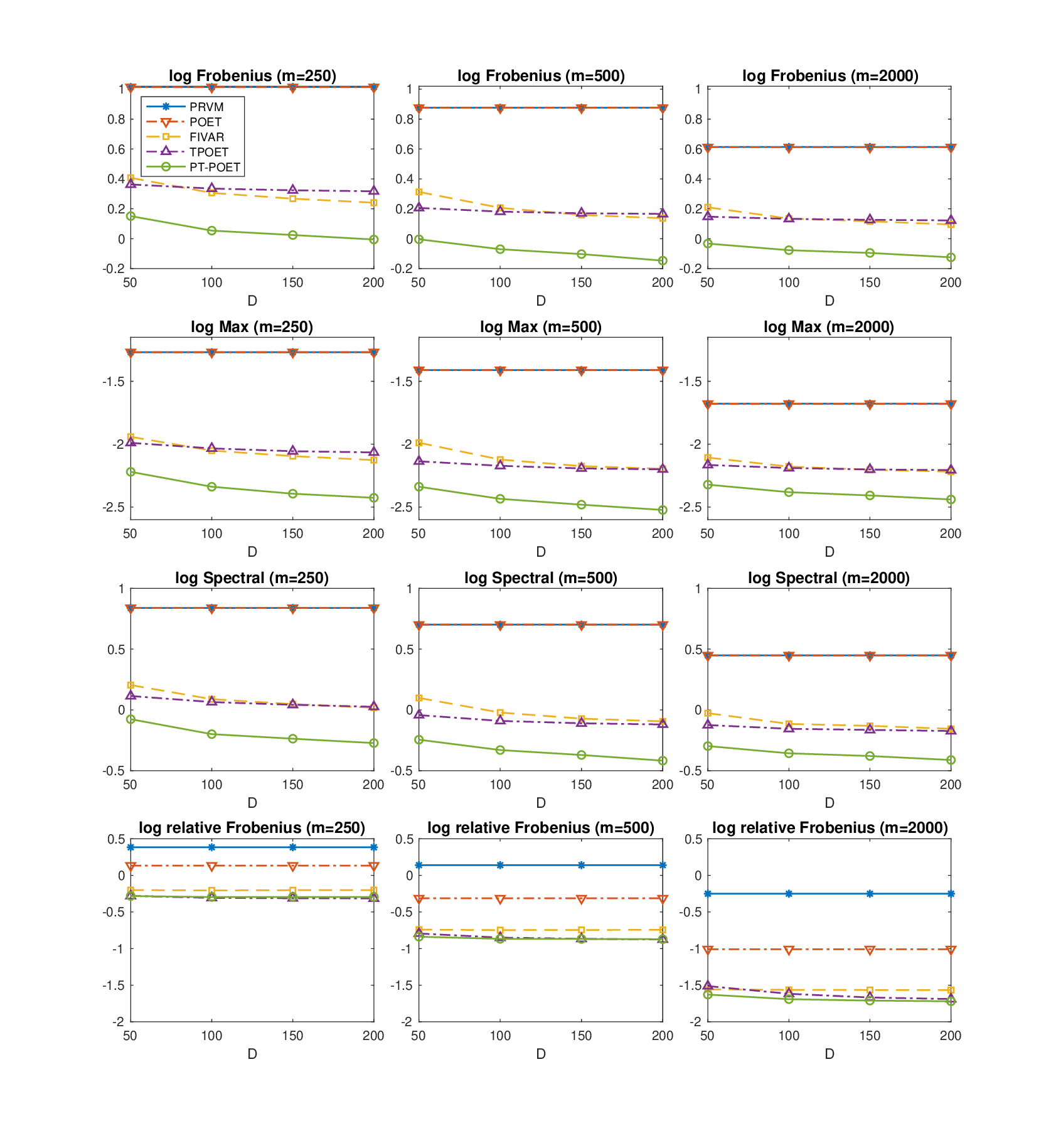}
	\centering	
	\caption{The log Frobenius, Max, Spectral, and relative Frobenius norm error plots of the PRVM, POET, FIVAR, T-POET, PT-POET estimators for the conditional expected integrated volatility matrix estimation against $D=50,100,150,200$, given $m = 250, 500, 2000$.}				\label{sim_plot}
\end{figure}

\section{Empirical Study} \label{empiric}
We applied the proposed PT-POET method to large volatility matrix prediction using real high-frequency trading data for 200 assets from January 2018 to December 2019 (503 trading days).
We selected the top 200 large trading volume stocks in the S\&P 500 in the Wharton Research Data Services (WRDS) system.
We used the previous tick scheme \citep{andersen2003modeling, barndorff2011multivariate, zhang2011estimating} to synchronize the high-frequency data to avoid the irregular observation time error issue, and we chose 1-min log-returns. 

We need to choose the ranks $ r_1$ and $ r_2$ to employ the proposed estimation procedure and other comparison methods. 
We first calculated 503 daily integrated volatility matrices using the PRVM estimation method in \eqref{PRVM}.
Then, we estimated the rank $r_1$ using the procedure suggested by \cite{ait2017using} as follows:
\begin{equation}
\hat{r}_{1} = \arg \min_{1 \leq j \leq r_{\max}} \sum_{d=1}^{503} \left[ p^{-1} \hat{\xi}_{d,j} + j \times c_1 \left\{ \sqrt{\frac{\log p}{m^{1/2}} + p^{-1} \log p} \right\}^{c_2} \right] - 1,
\end{equation}
where $\hat{\xi}_{d,j}$ is the $j$-th largest eigenvalue of PRVM estimator, $r_{\max} = 20$, $c_1 = 0.15 \times \hat{\xi}_{d,20}$, and $c_2 = 0.5$. The above method suggests $\hat{r}_{1} = 3$. 
In addition, Figure \ref{scree} represents the scree plot using the first 50 eigenvalues of the sum of 503 PRVM estimates.
Figure \ref{scree} confirms that this choice is reasonable. 
To estimate the rank $r_{2}$, we employed the largest singular value gap method as discussed in Section \ref{tuning}.
From those results, we set $r_{1}=3$ and $r_{2} =1$ for the empirical study.

To predict the conditional expected volatility matrix $E(\bGamma_{D+1} | \cI_{D})$, we employed the PT-POET, POET, FIVAR, and T-POET methods as described in Section \ref{simulation}.
Additionally, for robustness checks, we included PT-POET2, which is based on $r_{1}=3$ and $r_2 = 2$.
We also considered FIVAR-H, which incorporates a HAR structure instead of a VAR structure to predict one-day-ahead leading eigenvalues.
For PT-POETs,  we utilized the ex-post daily, weakly, and monthly realized volatility based on the first eigenvalues as covariates for $\bX$ and used the additive polynomial basis and $J =2$.
We used the rolling window scheme, where the in-sample period was one of 63, 126, or 252 days.
We presented the best-performing results for FIVAR, T-POET, and PT-POETs across the range of in-sample periods.
Specifically, the FIVAR method, utilizing observations from the past 252 days observations to estimate the model parameters and the previous 21 days to estimate the eigenvectors, performed best.
Following \cite{shin2021factor}, the AR lag $h=1$ was chosen based on the Bayesian information criterion (BIC).
For PT-POET and T-POET, an in-sample period of 63 days yielded the best performance.
We used three different out-of-sample periods: from 2019:1 to 2019:6 (period 1), from 2019:7 to 2019:12 (period 2), and from 2019:1 to 2019:12 (period 3).

\begin{figure}
	\includegraphics[width=0.8\textwidth]{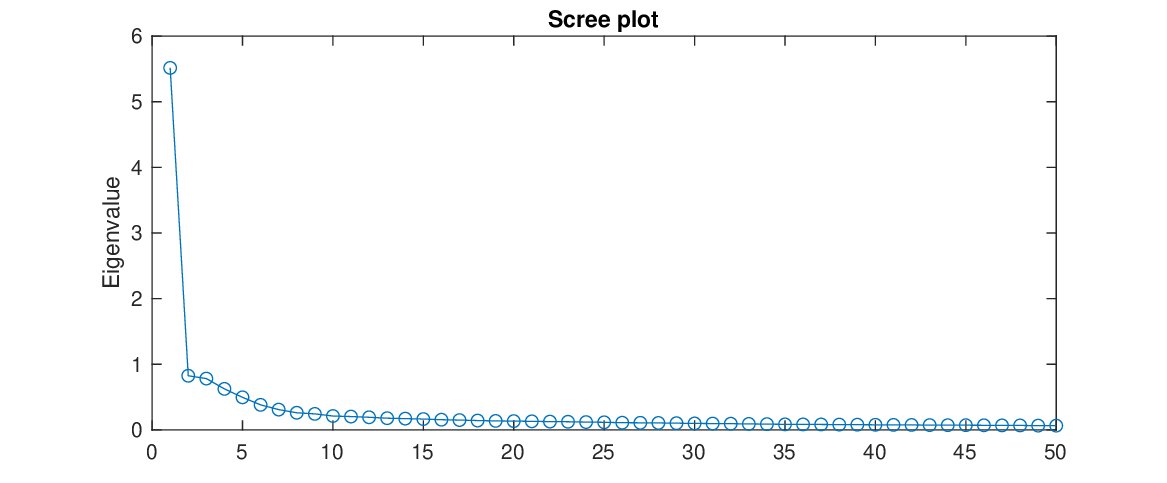}
	\centering	
	\caption{The scree plot of the first 50 eigenvalues of the sum of 503 PRVM estimates.}				\label{scree}
\end{figure}

For all estimators except PRVM, we estimated the idiosyncratic volatility matrix using the hard thresholding scheme based on the 11 Global Industrial Classification Standard (GICS) sectors \citep{ait2017using, fan2016incorporating}.
Specifically, we set the idiosyncratic components to zero across different sectors while retaining them within the same sector.

\begin{table}[h!]
    \centering
    \caption{MSPEs and QLIKEs for the PRVM, POET, FIVAR, FIVAR-H, T-POET, and PT-POETs.} \label{MSPE and QLIKE}
    \begin{tabular}{llllllll}
        \hline
        & PRVM & POET & FIVAR & FIVAR-H &  T-POET & PT-POET & PT-POET2 \\
        \midrule
        \multicolumn{8}{c}{MSPE$\times 10^4$} \\
        \cmidrule(lr){2-8}
        Period 1		&	1.300	&	1.269	&	0.942	& 0.971 &	0.951 &	0.878 & 0.880 \\
        % & \footnotesize(0.102) & \footnotesize(0.102) & \footnotesize(0.078) & \footnotesize(0.089) & \footnotesize(0.075) & \footnotesize(0.058) & \footnotesize(0.056)\\
        Period 2		&	3.184	&	3.163	&	2.035	& 2.047 &	2.704 &	2.025 & 2.039 \\
        %& \footnotesize(6.222) & \footnotesize(6.250) & \footnotesize(3.710) & \footnotesize(4.019) & \footnotesize(5.405) & \footnotesize(3.933) & \footnotesize(3.989)\\
        Period 3		&	2.242	&	2.216	&	1.488	& 1.509 & 	1.827 &	1.452 & 1.459 \\
        %& \footnotesize(3.238) & \footnotesize(3.253) & \footnotesize(1.861) & \footnotesize(2.075) & \footnotesize(2.806) & \footnotesize(2.020) & \footnotesize(2.043)\\
         \midrule
        \multicolumn{8}{c}{QLIKE$\times 10^{-3}$} \\
        \cmidrule(lr){2-8}
        Period 1		&	--	&	0.535	&	-0.030	& -0.031 &	8.343 &	-0.050 & -0.044 \\
        Period 2		&	--	&	0.574	&	-0.024	& -0.024 &	10.184 &	-0.041 &	-0.041\\
        Period 3		&	--	&	0.554	&	-0.027	& -0.027 & 	9.264 &	-0.046 &	-0.043 \\	
        \bottomrule
    \end{tabular}

\end{table}

To measure the performance of the predicted volatility matrix, we first utilized the mean squared prediction error (MSPE) and QLIKE \citep{patton2011volatility}:
\begin{align*}
&\text{MSPE} = \frac{1}{T} \sum_{d=1}^{T} \|\widetilde{\bGamma}_d - \hat{\bGamma}_d^{\text{POET}}\|_F^2, \\
&\text{QLIKE} = \frac{1}{T} \sum_{d=1}^{T} \log \left( \det \left( \widetilde{\bGamma}_d \right) \right) + \text{tr} \left( \widetilde{\bGamma}_d^{-1} \hat{\bGamma}_d^{\text{POET}} \right),
\end{align*}
where $T$ is the number of days in the out-of-sample period, $\hat{\bGamma}_d^{\text{POET}}$ is the POET estimator for the $d$-th day, which is a proxy of true volatility matrix, and $\widetilde{\bGamma}_d$ is one of the one-day-ahead volatility matrix estimates from PRVM, POET, FIVAR, FIVAR-H, T-POET, PT-POET, and PT-POET2 for the $d$-th day of the out-of-sample period.
In addition, since the true conditional expected large volatility matrix is unknown, we conducted the Diebold and Mariano (DM) test \citep{diebold2002comparing} using MSPE and QLIKE to assess the significance of differences in predictive performance.
We compared the proposed PT-POET method with other methods.
Table  \ref{MSPE and QLIKE} reports the results of MSPEs and QLIKEs, and Table \ref{DM_tests} shows the $p$-values for the DM tests.
We note that the QLIKE results for the PRVM estimator are omitted, as its determinant is close to zero.
The results indicate that the PT-POET estimators demonstrate the best overall performance, and PT-POET and PT-POET2 show statistically similar performances. 
This may be because incorporating the tensor structure and projection method, along with additional covariates such as ex-post realized volatility information, enhances prediction accuracy.
We note that PT-POET does not statistically outperform FIVAR-H based on the DM test using MSPE. 
This outcome is due to the higher variance of prediction errors with FIVAR-H compared to FIVAR and PT-POET.
That is, FIVAR-H is relatively volatile.  
To further evaluate the proposed method's performance, we implemented the minimum variance portfolio allocation, as described below.

\begin{table}[h!]
    \centering
        \caption{The $p$-values for the DM test statistic based on MSPE and QLIKE for the PRVM, POET, FIVAR, FIVAR-H, T-POET, and PT-POET2 with respect to the PT-POET.}
        \label{DM_tests}
    \begin{threeparttable}
        
        \begin{tabular}{lllllll}
            \toprule
            & PRVM & POET & FIVAR & FIVAR-H & T-POET & PT-POET2 \\
            \midrule
            \multicolumn{7}{c}{MSPE} \\
            \cmidrule(lr){2-7}
            Period 1 & 0.000*** & 0.000*** & 0.035** & 0.242 & 0.012** & 0.845\\
            Period 2 & 0.004*** & 0.005*** & 0.902  & 0.810 & 0.044** &0.265\\
            Period 3 & 0.001*** & 0.001*** & 0.062* & 0.372 & 0.043** & 0.284 \\
            \midrule
            \multicolumn{7}{c}{QLIKE} \\
            \cmidrule(lr){2-7}
            Period 1 & -- & 0.000*** & 0.000*** & 0.000***& 0.000*** &0.124 \\
            Period 2 & -- & 0.000*** & 0.005*** & 0.005*** & 0.000*** & 0.102 \\
            Period 3 & -- & 0.000*** & 0.000*** & 0.000** & 0.000*** & 0.190 \\
            \bottomrule
        \end{tabular}
        
    \end{threeparttable}
    
    % Manually position the footnote to span the full width
    \begin{minipage}{\textwidth}
        \footnotesize
        \text{Note:} ***, **, and * indicate rejection of the null hypothesis at significance levels of 1\%, 5\%, and 10\%, respectively.
    \end{minipage}

\end{table}

To analyze the out-of-sample portfolio allocation performance, we also considered the following constrained minimum variance portfolio allocation problem \citep{fan2012vast}:
\begin{equation*}\label{min problem}
	\min_{\omega} \omega^{T}\widetilde{\bGamma}_{d}\omega, \text{ subject to } \omega^{\top}\mathbf{1} = 1, \; \|\omega\|_{1}\leq c,
\end{equation*}
where $\mathbf{1} = (1,\dots,1)^{\top} \in \mathbb{R}^{p}$, the gross exposure constraint $c$ varies from 1 to 3, and $\widetilde{\bGamma}_{d}$ is one of the one-day-ahead volatility matrix estimators obtained from PRVM, POET, FIVAR, FIVAR-H, T-POET, and PT-POETs.
At the beginning of each trading day, we obtained optimal portfolios based on each estimator and held these portfolios for one day.
We calculated the realized volatility using the 10-min portfolio log-returns to avoid the microstructural noise effect.
We then measured the out-of-sample risk by averaging the square root of the realized volatility for each out-of-sample period.
Figure \ref{plot2019} illustrates the out-of-sample risks of the portfolios constructed by the PRVM, POET, FIVAR, FIVAR-H, T-POET, and PT-POET estimators.
As shown in Figure \ref{plot2019}, the PT-POET and PT-POET2 estimators demonstrate stable performance and consistently outperform the other estimators.
Interestingly, while T-POET performs well in terms of global minimum risk, it becomes unstable as the gross exposure constraint increases. 
Additionally, PRVM and POET do not perform well.
This may be because they cannot capture the dynamics of the volatility process.
By considering the vector auto-regressive structure on eigenvalues of the volatility matrix, FIVAR shows improved performance compared to POET.
However, FIVAR underperforms PT-POET because it does not consider the time-varying eigenvector.
When comparing PT-POET and PT-POET2, PT-POET2 shows a more stable performance.
This may be because adding an additional dimension of time series helps account for the eigenvector dynamics.
Overall, these results indicate that the PT-POET method can efficiently predict the large integrated volatility matrix by incorporating interday time series dynamics based on both time-varying eigenvector and eigenvalue structures.

\begin{figure}
	\includegraphics[width=\linewidth]{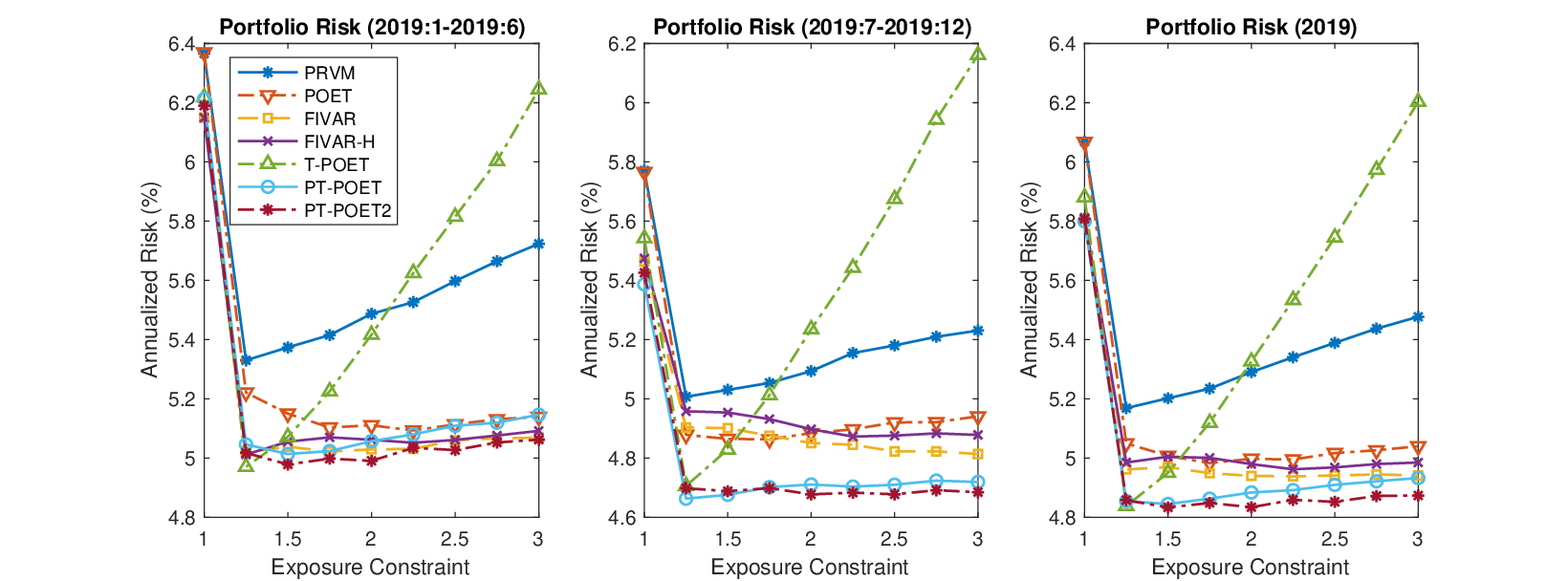}
	\centering	
	\caption{The out-of-sample risks of the minimum variance portfolios constructed by the PRVM, POET, FIVAR, FIVAR-H, T-POET, PT-POET estimators.}				\label{plot2019}
\end{figure}

\section{Conclusion} \label{conclusion}
This paper introduces a novel procedure for predicting large integrated volatility matrices using high-frequency financial data. 
The proposed  PT-POET method leverages daily volatility dynamics based on the semiparametric structure of the low-rank tensor component of the integrated volatility matrix process. 
We establish the asymptotic properties of PT-POET and its estimator for the future integrated volatility matrix.

In the empirical study, PT-POET outperforms conventional methods in terms of out-of-sample performance for predicting the one-day-ahead integrated volatility matrix and portfolio allocation. 
This finding confirms that generalizing the low-rank structure and incorporating the HAR model structure into interday volatility dynamics improves the prediction of future volatility matrices.
We note that, in this paper, we developed a generalized model for large volatility matrix processes and utilized the past leading eigenvalues as covariates for projecting the dynamics of singular vectors.
However, exploring other possible covariates, such as trading volume, with alternative basis functions would be an interesting direction for future research.
This study requires extensive empirical research.
Thus, we leave this for future research.

%This paper focuses on the interday volatility matrix process based on high-dimensional factor-based Itô processes. 
%In practice, forecasting the intraday large volatility matrix is essential for hedging, risk management, and portfolio management. 
%However, to achieve this, we would need to include an additional dimension for the intraday process and account for intraday periodic patterns, which makes it hard to develop prediction models.
%Thus, we leave this for a future research.

\bibliography{PTPOET}

	\end{document}